\documentclass[letterpaper]{article}
\usepackage{times}
\usepackage{helvet}
\usepackage{courier}
\usepackage{amssymb}
\usepackage{amsmath}

\usepackage{graphicx}
\usepackage{color}

\usepackage{lineno}

\usepackage{hyperref}
\usepackage{url}

\usepackage{algorithm}
\usepackage[noend]{algpseudocode}

\newtheorem{cor}{Corollary}

\newtheorem{exmp}{Example}
\newtheorem{example}[exmp]{Example}

\newtheorem{thm}{Theorem}

\newenvironment{proof}{{\bf Proof }}{\hfill$\Box$\par}

\newcommand{\given}{\,|\,}

\newcommand{\footremember}[2]{%
    \footnote{#2}
    \newcounter{#1}
    \setcounter{#1}{\value{footnote}}%
}
\newcommand{\footrecall}[1]{%
    \footnotemark[\value{#1}]%
} 
\providecommand{\keywords}[1]{\textbf{\textit{Keywords:}} #1}

\begin{document}
%
\title{Optimal Threshold Policies for Robust Data Center Control}
\author{%
	Paul Weng\footnote{SYSU-CMU Joint Institute of Engineering; SEIT, SYSU; SYSU-CMU Joint Research Institute}%
	 \and Zeqi Qiu\footremember{cmu}{Department of ECE, Carnegie Mellon University} %
	 \and John Costanzo\footrecall{cmu} %
	 \and Xiaoqi Yin\footrecall{cmu} %
	 \and Bruno Sinopoli\footrecall{cmu}%
}
\maketitle
\begin{abstract}
\begin{quote}
With the simultaneous rise of energy costs and demand for cloud computing, efficient control of data centers becomes crucial.
In the data center control problem, one needs to plan at every time step how many servers to switch on or off in order to meet stochastic job arrivals while trying to minimize electricity consumption.
This problem becomes particularly challenging when servers can be of various types and jobs from different classes can only be served by certain types of server, as it is often the case in real data centers.
We model this problem as a robust Markov Decision Process (i.e., the transition function is not assumed to be known precisely). 
We give sufficient conditions (which seem to be reasonable and satisfied in practice) guaranteeing that an optimal threshold policy exists.
This property can then be exploited in the design of an efficient solving method, which we provide.
Finally, we present some experimental results demonstrating the practicability of our approach and compare with a previous related approach based on model predictive control.
\end{quote}
\noindent \keywords{data center control, Markov decision process, threshold policy, robustness}
\end{abstract}

\section{Introduction}


The unprecedented growth of cloud computing has boosted the proliferation of data centers throughout the world. 
This increase is having a non negligible impact on energy consumption, business operating costs and the environment.
For instance, the Natural Resources Defense Council in the US expects that the electricity consumption due to data centers will raise to about 140 billion kilowatt-hours per year by 2020, which is equivalent to the annual output of 50 power plants, and would cost American businesses \$13 billion per year in electricity bills and would emit nearly 150 million metric tons of carbon pollution annually. 
Despite the importance of efficient data center control, many studies (e.g., \cite{KaplanForestKindler08}) report that up to 30\% of servers are comatose (i.e., idle servers consuming electricity, but providing limited or no useful services).

As a consequence, controlling data centers in an efficient way is an important problem as even little improvements in energy efficiency translate into significant savings. 
Such considerations have motivated a number of industry and academic researchers to investigate methods for reducing data center energy consumption, e.g., in electrical engineering and control \cite{MastroleonBambosKozyrakisEconomou05,ParoliniSinopoliKrogh08,ParoliniSinopoliKroghWang11,YinSinopoli14}, in operations research \cite{FeinbergZhang14} or in artificial intelligence and machine learning \cite{LubinKephartDasParkes09,Bodik10,Gao14,GhasemiLubin15}.

In this paper, we model the data center control problem as a Markov Decision Process (MDP).
Few previous work also uses MDP for tackling this problem.
One of the main reasons is that such modeling leads to large-sized MDPs, which are then difficult to solve and beside the computational issue, would also require a large space for representing an optimal policy.
Mastroleon et al. \cite{MastroleonBambosKozyrakisEconomou05} studied a one-server version of the data center control problem considering the presence of a job buffer and thermal information in order to choose how many CPUs to use.
They investigated the structural properties of their problem and proposed heuristics for solving the resulting MDP.
Parolini et al. \cite{ParoliniSinopoliKrogh08} considered the problem of coordinated cooling and load management.
The obtained constrained MDP is solved using linear programming, thus limiting the size of the problems (in e.g., number of servers) that can be tackled. 
Our work is closely related to that of Feinberg and Zhang \cite{FeinbergZhang14}.
They formulated the data center control problem as an MDP with an infinite number of identical servers (to approximate large-sized MDP) and one job class. 
In that formulation, they proved the existence of an optimal threshold policy. 

Related to Feinberg and Zhang's line of work, the existence of optimal threshold policies have been investigated in many domains, e.g., in inventory control \cite{HordijkSchouten86,Kalin80}, in scheduling \cite{HyonJeanMarie12}, in battery control \cite{FoxLongMagazzeni11,PetrikWu15}, but also in general settings \cite{ErsegheZanellaCodemo13,Koole07}.
Most of those works assume a perfect knowledge of the stochastic environment and deal with the mono-dimensional case (e.g., one type of product or one type of server).
Only a few consider the multi-dimensional case.
In inventory control, Veinott \cite{Veinott65} studied the multi-product inventory problem.
In queueing problems, Koole \cite{Koole07} studied the bi-dimensional case. 

As underlined by Yin and Sinopoli \cite{YinSinopoli14}, real data centers have heterogeneous servers located at different locations and jobs may be of different classes.
Previous work based on MDPs cannot handle such more realistic cases.
Besides, they assume a perfect knowledge of the MDP transition function, which is clearly a very strong assumption in the domain of cloud computing where job arrival rates may change and are difficult to predict \cite{JuanLiPengMarculescuFaloutsos14}.
MDPs have been extended to handle model uncertainty \cite{GivanLeachDean00,NilimElGhaoui03}.
Our modeling of the data center control, based on robust MDPs as studied by Nilim and El Ghaoui \cite{NilimElGhaoui03}, allows for heterogeneous servers organized into server blocks (i.e., identical servers at a same location), different job classes and constraints on which job class can be handled by which server block.
We prove under some sufficient natural conditions the existence of optimal threshold policies in the multi-dimensional setting (i.e., many server types and blocks and many job classes).

Such results are important for several reasons.
This allows to design more efficient solving methods (as one only needs to search for an optimal threshold policy).
We propose an adaptation of backward induction adapted to our model.
Moreover, an optimal policy has a more compact representation.
This would facilitate its application in practice.
Finally, a threshold policy is interpretable and thus easily understandable for humans.
To the best of our knowledge, this is the first such result under model uncertainty and the first in the data center control problem in the multi-dimensional setting.


\section{Background and Problem Definition}
In this section, we first recall the definition of a Markov Decision Process and then use it to model the data center control problem.

\subsection{Markov Decision Process}

A {\em Markov Decision Process} (MDP) \cite{Puterman94} is defined as $\langle S, A, (p_t)_{t\ge 1}, (c_t)_{t\ge 1} \rangle$ where $S$ is a finite set of states, $A$ is a finite set of actions, $p_t(s, a, s')$ is the transition probability of reaching state $s'$ by executing action $a$ in state $s$ at time-step $t$ and $c_t(s, a)$ is the cost incurred by applying action $a$ in state $s$ at time-step $t$.
In the finite horizon case (i.e., finite number of decisions), we allow the transitions probabilities and costs to be non-stationary.
In the infinite horizon case, we will assume that they are stationary (i.e., independent of time step).

A {\em policy} $\pi = (\delta_1, \delta_2, \ldots)$ is a sequence of decision rules $\delta_t : S \to A$. 
It is said to be {\em stationary} if the same decision rule is applied at every time step.
In that case, we identify the stationary policy with its decision rule.

At horizon $h$, the {\em value function} $v^\pi_1$ of a policy $\pi = (\delta_1$, $\delta_2, \ldots$, $\delta_h)$ is defined by: $\forall s\in S$
\begin{align}
v^\pi_{h+1}(s) &= 0\\
v^\pi_t(s) &= c_t\big(s, \delta_t(s)) + \gamma \mathbb E_{s' \sim p_t(s, \delta_t(s), \cdot)}[v^\pi_{t+1}(s')]\label{eq:vpi}
\end{align}
where $\gamma\in [0, 1)$ is a discount factor.
At horizon $h$, the value function $v^*_1$ of an optimal policy can be computed as follows: $\forall s\in S$,
\begin{align}
v^*_{h+1}(s) &= 0\\
v^*_t(s) &= \min_a c_t(s, a) + \gamma \mathbb E_{s' \sim p_t(s, a, \cdot)}[v^*_{t+1}(s')]\label{eq:opt}
\end{align}

At the infinite horizon, it is known that there exists an optimal stationary policy (assuming the MDP is stationary, i.e., $p_t = p$ and $c_t = c$ are constant).
The value of a stationary policy $\pi$, which is defined as the limit of (\ref{eq:vpi}) when $h$ tends to infinity, satisfies: $\forall s\in S$
\begin{align}
v^\pi(s) & = c(s, \pi(s)) + \gamma \mathbb E_{s' \sim p(s, \pi(s), \cdot)}[v^\pi(s')]
\end{align}
At the infinite horizon, an optimal policy is a policy that satisfies the Bellman equations: $\forall s\in S$
\begin{align}
v^*(s) &= \min_a c(s, a) + \gamma \mathbb E_{s' \sim p(s, a, \cdot)}[v^\pi(s')]
\end{align}
Although directly solving the Bellman equations (using the value iteration algorithm for instance) leads to a polynomial-time (in $|S|, |A|, 1/(1-\gamma)$) complexity solving method \cite{LittmanDeanKaelbling95}, in practice this may not be feasible when the number of states is large.
In such cases, one needs to exploit the structure of the problem for designing more efficient algorithms.
We will show how this can be done for the data center control problem, which we present now.

\subsection{Data Center Control}

The data center control problem that we tackle in this work is inspired by the formulation of Yin and Sinopoli \cite{YinSinopoli14}.
We first describe this problem and then model it as a Markov Decision Process.

We start by introducing some notations.
For any integer $n$, $[n]$ denotes the set $\{1, \ldots, n\}$.
For any real $z$, $|z|_+$ denotes $\max(z, 0)$.
For two real vectors $x, x'$ of the same dimension, $x \cdot x'$ denotes their inner product and any matrix $M$ and vector $x$, $M\cdot x$ denotes their matrix product.
Any discrete function $f : \mathbb N \to \mathbb R$ can be extended to $\mathbb R$ by considering its linear interpolation.
Such discrete function $f$ is said to be convex if its linear interpolation is convex.

\begin{figure*}
\centering
\scalebox{.5}{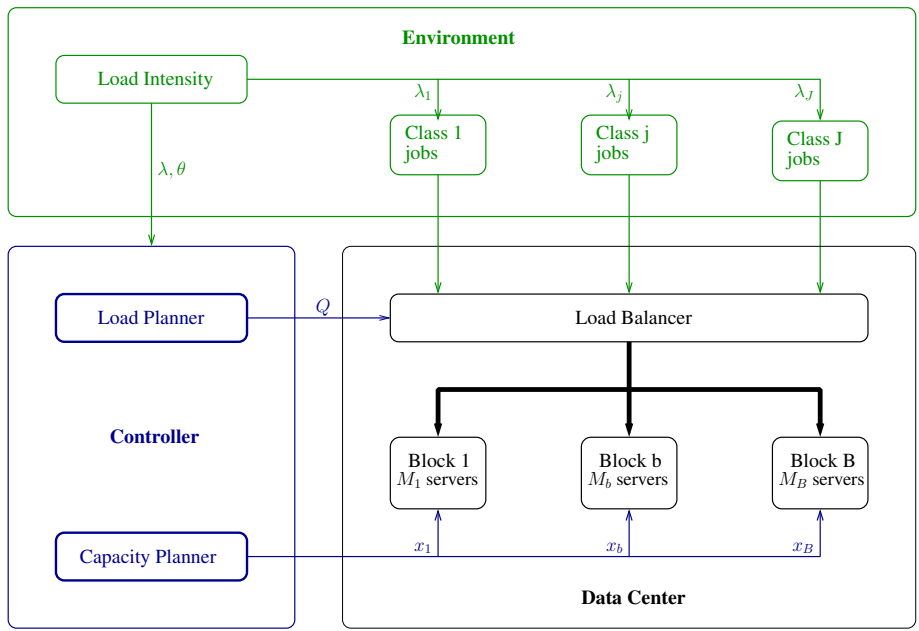}
\caption{Data Center Control}
\label{fig:dcc}
\end{figure*}

In the data center control problem (Figure~\ref{fig:dcc}) with heterogeneous servers and multiple classes of jobs, servers of $I$ different types are controlled (i.e., switched on or off) in order to serve jobs that may be of $J$ different classes and whose job arrivals are stochastic. 
Servers are arranged in $B$ blocks.
A server block represents a group of identical servers that are in the same location.
The set of server blocks is identified to $[B]$ and block $b \in [B]$ contains $M_b$ servers of the same type.
The set of job classes is represented by $[J]$ and the set of server types is identified to $[I]$.
There may be restrictions on which server type can serve which job class.
Without loss of generality, we assume that a server block only contains identical servers that can serve jobs of the same job classes.
As two different server blocks may contain servers of the same type, thus the number of server types is smaller or equal to the number of server blocks, i.e., $I\le B$
We assume time is discretized and that during one time slot (i.e., duration between two time steps), all jobs are served (or some may be dropped).

\begin{example}
Wikipedia manages a certain number of servers, each specialized in one type of media (text, image, video...).
Here, a job is a request for a resource from Wikipedia and a job class is a type of media.
A server block regroups servers in a same location, serving the same type of media.
Here, the assumption that all jobs are served between one time slot is reasonable as the retrieving and sending a Wikipedia resource should be very fast.
\end{example}

Moreover, the stochastic arrival rates of class-$j$ jobs are assumed to be:
\begin{itemize}
\item observable,
\item distributed according a parametrized probability distribution $P_t(\lambda \given \theta)$, where $\lambda = (\lambda_1, \ldots, \lambda_J)$ is the vector of arrival rates for each job class and $\theta$ belongs to a finite set $\Theta$ representing different level of load intensities (e.g., peak-hours, normal, off-hours) and the dynamics of the parameter $\theta$ is governed by a Markov model $P_t(\theta' \given \theta)$.
Therefore, the distribution over next arrival rates is given by $P_t(\lambda_j' \given \theta') P_t(\theta' \given \theta)$.
\end{itemize}
At every time step, given the current arrival rate, one has to make two decisions:
\begin{description}
\item[Load Balancing Planner:] how to do the load balancing over server blocks.
Indeed, for each job of a certain class that arrives, one needs to decide to which server block it is dispatched to.
The load balancing decision problem can be seen as choosing a load balancing matrix $Q$ of size $B\times J$ where $Q_{bj}$ represents the proportion of jobs of class $j$ to be sent to server block $b$.
Note that $\sum_{b\in[B]} Q_{bj} = 1$ for any job class $j$.
We denote $\Omega_Q \subset \mathcal M_{B \times J}(\mathbb R)$ the set of all load balancing matrices that are allowed (they may encode constraints, such as which type of job can be served by which server block).
\item[Capacity Planner:] how many servers in each server block to turn on or off for the next time step. 
We assume the time to switch on or off is negligible with respect to the duration between two time steps.
\end{description}
The ultimate goal is to minimize the expected total energy consumption cost and QoS (Quality of Service) cost.
All costs are assumed to be expressed in the same unit (e.g., dollars) and can therefore be summed.
During one time slot $t$, the energy consumption cost $c^E_t(x, x')$ is defined by
\begin{align}
c^E_t(x, x') =& \sum_{b \in [B]} e_b(t, x_b) + c^{\pm}_t(x, x')\\
c^{\pm}_t(x, x') =& \sum_{b\in [B]} c^+_b(t) |x'_b - x_b|_+ + c^-_b(t) |x_b - x'_b|_+ \label{eq:onoff}
\end{align}
where $x = (x_1, \ldots, x_B)$ (resp. $x' = (x'_1, \ldots, x'_B)$) represents the number of servers that are on in each server block at the current (resp. next) time step (each $x_b$ and $x'_b$ are integers between $0$ and $M_b$),
$e_b(t, x_b)$ is the energy consumption of $x_b$ servers in block $b$ during time slot $t$,
and $c^+_b(t) \in \mathbb R^+$ (resp. $c^-_b(t) \in \mathbb R^+$) is the cost of switching on (resp. off) one server in server block $b$ during time slot $t$.
Both $e_b$, $c^+_b$ and $c^-_b$ depend on $t$, which allows time-varying electricity prices.
In its simplest form, $e_b(t, x_b)$ could be defined as $E_b(t) x_b$, where $E_b(t) \in \mathbb R^+$ is the electricity consumption cost due to leaving one server on in block $b$ during one time slot.
More generally, $e_b(t, x_b)$ could also integrate other costs, such as for instance, those related to the cooling of the server block, which would also depend on the number of servers on. 

Following \cite{YinSinopoli14}, the QoS cost $c^{QoS}_j(x, \lambda, Q)$ related to job class $j$ during one time slot can be obtained as follows:
\begin{align}
&c^{QoS}_j(x, \lambda, Q) 
= C_j \lambda_j d_j \label{eq:cj}\\
&d = Q^t \cdot d^S \label{eq:resp}\\
&d^S_b(x_b, \lambda^S_b) = \frac{x_b}{r_b x_b - \lambda^S_b} \mbox{ with } r_b x_b > \lambda^S_b \label{eq:respS}\\
&\lambda^S = Q \cdot \lambda \label{eq:rateS}
\end{align}
where $r_b$ is the processing rate of a computer in server block $b$.
Equation~(\ref{eq:rateS}) computes the arrival rates in each server block.
Equation~(\ref{eq:respS}) defines the mean response time in server block $b$ for all jobs.
Note that $Q$ needs to be chosen such that the condition of (\ref{eq:respS}) is satisfied.
Equation~(\ref{eq:resp}) represents the mean response time for each job class at the data center level, i.e., $d=(d_1, \ldots, d_J)$.
Finally (\ref{eq:cj}) defines the QoS cost for job class $j$, where $C_j \in \mathbb R^+$ is a parameter.

The total QoS cost during one time slot is thus:
\begin{align}
c^{QoS}(x, \lambda, Q) = \sum_{j \in [J]} c^{QoS}_j(x, \lambda, Q)
\end{align}
and the total cost incurred during one time slot is simply the sum of the energy consumption cost $c^E_t(x, x')$ and the total QoS cost $c^{QoS}(x, \lambda, Q)$.

The data center control problem as presented above is slightly more general than the one tackled in \cite{YinSinopoli14}.
They solved that problem with a two-stage approach based on the adaptive robust optimization framework.
They assume a prediction of the future arrival rates under the form of intervals is given.
Based on this prediction, the first stage computes the optimal capacity (i.e., number of servers on) of each server block.
Then, in the second stage, given the capacity plan, an optimal load balancing can be determined based on the observed arrival rates.
This problem can be viewed as a Markov Decision Process (MDP), which we present below, and the two-stage approach can be understood as an approximation solving method for that MDP.

A direct modeling of the data center control problem as a Markov decision process (MDP) yields $\langle S, A, (p_t)_{t\ge 1}, (c_t)_{t\ge 1} \rangle$:
\begin{itemize}
  \item $S = \Omega_x \times \Lambda \times \Theta$ is a set of states, where $\Omega_x = \prod_{b\in [B]}[M_b]$ with $M_b$ the number of servers in server block $b$, $\Lambda \subseteq \mathbb R^J_+$ denotes the set of all possible arrival rate vectors and $\Theta$ is the set of probability distribution (over $\Lambda$) parameters (e.g., representing peak-hours, off-peak hours...),
  \item $A = \Omega_Q \times \Omega_x$ is a set of actions, where $\Omega_Q \subset \mathcal M_{B \times J}(\mathbb R)$ is the set of all load balancing matrices (they may encode constraints, such as which job class can be served by which server block) and an action represents the two decisions that need to be made: a load balancing matrix and a vector of numbers of servers that will be on in each block during the next time slot,
  \item $p_t : S\times A\times S\to [0, 1]$ is a transition function defined as follows: $\forall s=(x, \lambda, \theta) \in S, s'=(x', \lambda', \theta') \in S$,
  \begin{align} \label{eq:transition}
  p_t\big(s, (Q, a), s'\big) &= \mathbb I(a=x') P_t(\lambda' | \theta') P_t(\theta' | \theta)
  \end{align}
  where $\mathbb I$ is the indicator function.
  \item $c_t : S\times A\to \mathbb R$ is a cost function defined as follows:
   \begin{align}
  c_t\big((x,\lambda, \theta), &(Q, a)\big) = c^{QoS}(x, \lambda, Q) + c^E_t(x, a) 
  \end{align}
  which is the total cost incurred during one time slot.
  Note that the cost function does not depend on parameter $\theta$.
\end{itemize}
Our MDP formulation generalizes the formulation of previous work \cite{YinSinopoli14}.
Note that the state of the system contains the current number $x$ of servers on in each server block because it impacts the costs.
We assume that the transition function is time-dependent and has a special form (\ref{eq:transition}).
This implies that parameter $\theta$ evolves as a Markov chain and that the stochastic evolution of $\lambda$ depends on $\theta$, which allows to model various dynamics.
For simplicity, we assume that $x$ does not evolve stochastically.
Our theoretical results could easily be extended to the case where $x$ also changes stochastically (e.g., because of server failures).

Interestingly, this model can be simplified. 
Indeed, in the previous MDP model, one can notice that the best choice of load balancing matrix $Q$ depends only on the current arrival rate.
Therefore, the choice of $Q$ is completely independent among different time steps.
Then the problem can be reformulated as the following simpler model $\langle S, A, (p_t)_{t\ge 1}, (c_t)_{t\ge 1} \rangle$:
\begin{itemize}
  \item $S = \Omega_x \times \Lambda \times \Theta$, 
  \item $A = \Omega_x$, where an action represents a vector of numbers of servers that will be on in each blocks during the next time slot,
  \item $p_t : S\times A\times S\to [0, 1]$ is defined as before. 
  \item $c_t : S\times A\to \mathbb R$ is a cost function defined as follows:
   \begin{align}
  c_t\big((x,\lambda, \theta), &a\big) = c^{QoS}(x, \lambda) + c^E_t(x, a) 
  \end{align}
  where $c^{QoS}(x, \lambda)$ is the QoS cost when $x$ servers are on and arrival rates are $\lambda$. 
  It can be computed as the solution of an optimization problem:
\begin{align}
c^{QoS}(x, \lambda) = \min_Q \{ &\sum_{j \in \mathcal J} c_j^{QoS}(x, \lambda, Q) \mbox{ such that }\\
& Q \lambda \prec x \mbox{ and } Q \in \Omega_Q\}
\end{align}
where $\prec$ is the componentwise natural order.
\end{itemize}

\subsection{Threshold Policy}

In the literature considering threshold policies, generally only the mono-dimensional case is considered, i.e., for our problem, the number of server blocks would have to be one.
In that case, a {\em threshold decision rule} is defined by:
\begin{align}
\delta(x, \theta) &= \tau_1(\theta) \mbox{ if } x \le \tau_1(\theta)\\
&= \tau_2(\theta) \mbox{ if } x \ge \tau_2(\theta)\\
&= x \mbox{ otherwise}
\end{align}
where $\tau_1(\theta) \le \tau_2(\theta)$ are thresholds. 
We will show in Theorem~\ref{thm:finite} that those thresholds only depend on parameter $\theta$.
A policy composed of threshold decision rules is called a {\em threshold policy}.

While storing a decision rule in general requires $\mathcal O(|A|^{|S|})$ space, a threshold decision rule needs only $\mathcal O(2^{|\Theta|})$.
Note that in practice the size of $\Theta$ is generally very small.
Beside enjoying a compact representation, a threshold decision rule has a simple interpretation.
In the data center control problem, a threshold decision rule can be interpreted as follows: 
if the current number of servers on is too low, switch some servers on;
on the contrary, if there are too many servers on, switch some servers off;
otherwise do nothing.
The thresholds, which depend on parameter $\theta$, give some indications about the number of jobs that will arrive next time step (via $P_t(\cdot \given \theta)$).

Assuming that there is only one server block is not very realistic in modern data centers.
In this paper, we consider the case where there may be many server blocks.
We will show how the definition of a threshold policy can be extended in the multidimensional case.

\subsection {Robust MDP}

In practice, the transition function of an MDP is difficult to estimate and rarely known exactly.
Following \cite{NilimElGhaoui03}, an uncertain MDP is an MDP where each transition probabilities $p_t(s, a, \cdot)$'s lie in a subset $\mathcal P_t^{s, a}$ of the probability simplex $\{y \in [0, 1]^{|S|} \given \sum_i y_i = 1; \forall i, y_i\ge 0\}$.
The subset can be thought of as a confidence interval for the true transition probabilities.
As customary in robust control, the optimal value function for a finite horizon $h$ is redefined by: $\forall s\in S$,
\begin{align}
v^*_{h+1}(s) &= 0\\
v^*_t(s) &= \min_a \max_{p \in \mathcal P_t^{s, a}} c_t(s, a) + \gamma \mathbb E_{s' \sim p}[v^*_{t+1}(s')]\label{eq:optrobust}
\end{align}

In the data center control problem, this uncertainty concerns the Markov process governing the evolution of parameter $\theta$. 
We assume therefore that each $P_t(\cdot \given \theta)$ lie in a subset $\mathcal P_t^\theta$ of the probability simplex $\{y \in [0, 1]^{|\Theta|} \given \sum_i y_i = 1; \forall i, y_i\ge 0\}$.

\section{Optimality of Threshold Policies}
We now describe how a threshold policy can be defined for our MDP, then present a solving method exploiting the result obtained in Theorem~\ref{thm:finite}.
Recall a state $s\in S$ is a triplet $s = (x, \lambda, \theta)$.

\subsection{Threshold Policy in Multidimensional Case}

First, note that Equation~(\ref{eq:opt}) can be rewritten as:
\begin{align}
v^*_t\big(s
\big) & = \min_a \max_{P \in \mathcal P_t^\theta} \max_{k \in \mathcal K} f_t^{k,P}(x, \lambda, \theta, a) \label{eq:minimax}\\
&= \min_a \max_{k \in \mathcal K} \max_{P \in \mathcal P_t^\theta} f_t^{k,P}(x, \lambda, \theta, a)\\
&= \min_a \max_{k \in \mathcal K} f_t^{k}(x, \lambda, \theta, a)
\end{align}
where
\begin{align}
f_t^{k}(s &
, a) = \max_{P \in \mathcal P_t^\theta} f_t^{k,P}(x, \lambda, \theta, a) \\ 
f_t^{k,P}(s, \;&
a) = 
c_t\big((x, \lambda, \theta), a\big) \label{eq:f}\\
& + \gamma \mathbb E_{\theta' \sim P}\mathbb E_{\lambda' \sim P_t(\cdot \given \theta')}[v^*_{t+1}\big(a,\lambda',\theta'\big)] \nonumber
\end{align}
Note also that Equation~(\ref{eq:onoff}) can be rewritten as:
\begin{align}
c^{\pm}_t(x, x') &= \sum_{b\in [B]} \frac{1+k_b}{2} c^+_b (a_b - x_b) \\
& + \sum_{b\in [B]} \frac{1-k_b}{2} c^-_b (x_b-a_b) 
\end{align}
Function $f_t^{k}(s
, a)$ can be rewritten as $f_t^{k}(s
, a) = g_t^k(x, \lambda) + h_t^{k}(a, \theta)$ where:
\begin{align}
g_t^k(x, &\lambda) = 
c^{QoS}(x, \lambda) + \sum_{b \in [B]} e_b(t, x_b) \label{eq:g}\\
&+ \sum_{b\in [B]} x_b [ \frac{1-k_b}{2} c^-_b - \frac{1+k_b}{2} c^+_b ] \nonumber
\\
h_t^{k}(a, &\theta) = 
\sum_{b\in [B]} a_b [ \frac{1+k_b}{2} c^+_b - \frac{1-k_b}{2} c^-_b ]\label{eq:h} \\
&+\gamma \displaystyle\max_{P \in \mathcal P_t^\theta} \mathbb E_{\theta' \sim P} \mathbb E_{\lambda' \sim P_t(\cdot \given \theta')}[v^*_{t+1}\big(a,\lambda', \theta'\big)] \nonumber
\end{align}

The definition of a threshold decision rule can then be extended as follows. 
Let $\mathcal K = \{-1,1\}^B$.
For a given $x$, a vector $k \in \mathcal K$ represents an orthant of $\Omega_x$ centered at $x$.
A component $k_b$ of a vector $k$ indicates if cost $c^+_b$ (if $k_b=1$) or $c^-_b$ (if $k_b=-1$) is selected in (\ref{eq:f}).
In the general case, a threshold decision rule $\delta$ is written:
\begin{align}
\delta\big(x, \theta\big) &= \tau^k(\theta) \mbox{ if } \left\{
\begin{array}{l}
\exists k\in \mathcal K, \forall b\in [B], \\
k_b\cdot (\tau^{k}(\theta) - x_b) \ge 0 \label{eq:thresholdmulti}
\end{array}\right.\\
&= x \mbox{ otherwise}
\end{align}
where for all $k\in \mathcal K$, $\tau^k(\theta) \in \Omega_x$ are thresholds.
The thresholds $\tau^k$'s are defined such as the condition in (\ref{eq:thresholdmulti}) occurs at most for only one $k$.
The interpretation of such a decision rule is as follows:
if there is a $k$ such that $\forall b\in [B], k_b\cdot (\tau^{k}(\theta) - x_b) \ge 0$, 
switch on $\tau^{k}(\theta) - x_b$ servers in server block $b$ for any $k_b>0$ (i.e., $\tau^{k}(\theta) \ge x_b$) and
switch off $x_b - \tau^{k}(\theta)$ servers in server block $b$ for any $k_b<0$ (i.e., $\tau^{k}(\theta) \le x_b$)
otherwise, do nothing.

\subsection{Finite Horizon}

We show that in the finite horizon case, there exists an optimal policy that is a threshold policy.
This result can be extended to the infinite horizon case (see discussions section).

\begin{thm}\label{thm:finite}
If $c_j(x, \lambda, Q)$ is a convex function of $(x, Q)$ for all $\lambda$, then the robust MDP formulated above admits for any finite horizon $h$ an optimal policy that is a threshold policy.
\end{thm}
\begin{proof}
If $c_j(x, \lambda, Q)$ is convex in $(x, Q)$ for all $\lambda$, then $c^Q(x, \lambda)$ is convex in $x$ for all $\lambda$ as the minimum of a family (indexed by $Q$) of sums of convex functions.

We prove by induction on $t$ that $v_t^*(x, \lambda, \theta)$ is convex in $x$, and $f_t^{k}$'s are convex in $(x, a)$ for all $\lambda, \theta$.
Obviously, $v_{h+1}^*(x, \lambda, \theta)$ as a constant function is convex.
Functions $f_h^{k}$'s  are convex in $(x, a)$ as sums of a convex function and linear functions.
Assume that $v_{t+1}^*(x, \lambda, \theta)$ is convex in $x$, and $f_{t+1}^{k}$'s are convex in $(x, a)$ for all $\lambda, \theta$
Therefore, each $f_t^{k}$ is convex in $(x, a)$ as all the operations (sum, max, expectation) in its definition preserve convexity.
Then (\ref{eq:minimax}) implies $v_t^*(x, \lambda, \theta)$ is convex in $x$, since the pointwise maximum of convex functions is convex and the minimization in $a$ of a function convex in $(x, a)$ is convex in $x$.

We now prove that there is an optimal threshold policy.
Functions $f_t^k$ satisfy: $\forall s=(x, \lambda, \theta)$,
\begin{align}
f_t^k(s, x) &= f_t^{k'}(s, x) \mbox{ for } k, k' \in \mathcal K \label{eq:eq} \\
f_t^k(s, a) &\ge f_t^{k'}(s, a) \mbox{ for } \left\{
\begin{array}{l}
k, k' \in \mathcal K, \\
k\cdot (a - x) \ge 0 \label{eq:ge} 
\end{array}
\right.
\end{align}
Index $k$ selects an orthant in the space of actions centered at $x$.
The last equation implies that if the optimal action in (\ref{eq:minimax}) is in the orthant indexed by $k$ (i.e., $k\cdot (a - x) \ge 0$), the minimum is attained via $f_t^k$.
For a fixed $(x, \lambda)$, let $a^k$ be the action that minimises $f_t^k(x, \lambda, \theta, a)$.
Equation (\ref{eq:f}) shows that $a^k$ does not depend on $x$ and $\lambda$. It only depends on $\theta$ (and $t$).
Two cases can happen:
\begin{itemize}
\item There exists $k \cdot (a^k - x ) \ge 0$: 
By contradiction, assume $a^{k'}$ with $k'\neq k$ is the only solution of (\ref{eq:minimax}), which implies $f^{k'}_t(x, \lambda, \theta, a^{k'}) < f_t^k(x, \lambda, \theta, a^k)$.
By (\ref{eq:ge}), $f_t^{k'}(x, \lambda, \theta, a^{k'}) \ge f^{k}_t(x, \lambda, \theta, a^{k'})$, which leads to a contradiction.
Therefore, $a^k$ is solution of (\ref{eq:minimax}).
\item Otherwise, we have $\forall k, k \cdot (a^k - x ) < 0$. 
As $f_t^k$ is convex in $a$, $f_t^k(x, \lambda, \theta, a) \ge f_t^k(x, \lambda, \theta, x)$ for $k \cdot (a -  x) \ge 0$.
Therefore, $x$ is solution of (\ref{eq:minimax}).
\end{itemize}
\end{proof}

%

\subsection{Solving Method}



\begin{algorithm}[tbp]
 \caption{Backward-Induction$(\langle S, A, (p_t)_{t\ge 1}, (c_t)_{t\ge 1} \rangle$, $h, \gamma)$}
 \label{alg:bi}
 \begin{algorithmic}[1]
 \State $v^*_{h+1}\big((x, \lambda)\big) \gets 0$, $\forall x, \lambda\in\Omega_x\times\Lambda$
 
 \For{$t=h, h-1, \ldots, 1$}
   \For{$\theta\in\Theta$}
    \State compute threshold $\tau_t^k$ and optimal values $h^{*k}_t$ by solving $\min_a h_t^k(a, \theta)$ (see (\ref{eq:h})) for all $k$ \label{alg:1}
     \For{$(x, \lambda)\in\Omega_x\times\Lambda$}
       \State $v^*_t(x, \lambda, \theta) \gets \max_k g_t^{k}(x, \lambda) + h^{*k}_t$
     \EndFor
    \EndFor
  \EndFor
 \State \Return $v^*_1, \ldots, v^*_h$
 \end{algorithmic}
\end{algorithm}

An algorithm based on backward induction (see Algorithm~\ref{alg:bi}) can compute an optimal threshold policy and its value function.
In Line~\ref{alg:1}, the maximization in (\ref{eq:h}) can be performed efficiently using a bisection search \cite{NilimElGhaoui03}.
This operation only needs to be performed over the set of parameters $\Theta$, which should be small in practice.

As the algorithm contains a loop over $\Omega_x$, the overall time complexity is exponential in the number of server blocks.
Our approach can handle medium-sized data centers, but other techniques would be needed for bigger-sized data centers.
One such technique consists in applying our approach on a data center control model with more restrictive assumptions, which would serve as an approximation of the real problem.
We present two sets of stronger assumptions, defined in Section 5.2, yielding both a time complexity exponential in the number of server types, which may be reasonably small in practice, even for large-sized data centers.
They may potentially be used as a relaxation or approximation technique even if those assumptions are not satisfied in practice.

\section{Experimental Results}

We use publicly available page request history\footnote{\url{https://dumps.wikimedia.org/other/pageviews/2015/}} for Wikipedia. 
The dataset includes the total number of requests received every hour in 2015.
We consider four job classes corresponding to the different available resources in English and Spanish: en-html, en-image, es-html, es-image. 
We assume we have four server blocks with the following number of servers $[30,50,6,3]$, and that each handles a different class of job. 

We apply the following steps in order to estimate the transition function.
We use $k-$means to cluster the arrival rate vectors. 
Each cluster corresponds to a different value of the hidden parameter $\theta$. 
Here, we choose $k=3$ to represent three modes: high-peak, normal, low-peak.
From this, the emission probabilities are estimated.
The parameters of the QoS cost function are chosen as follows: $E_b=10$, $c_b^+=c_b^-=3$, $C_j=1$ and $r_1 = 610K, r_2 = 6.1K, r_3 = 610K, r_4 = 6.1K$.

We run the backward-induction algorithm at a horizon of $h = 24$ (e.g., one day) to obtain the threshold policies when the transition function is precisely known. 
On the same domain, we also run the MPC approach proposed by \cite{YinSinopoli14} where the uncertainty set of the arrival rates is set large enough to cover all possible transitions. 
Figure~\ref{fig:mdpmpc} shows the cumulated costs obtained by the MDP and MPC approaches for the last 12 hours (the costs start at zero and are therefore close for the two methods during the first 12 hours).
Dark lines represent the average obtained over 1000 runs.
Light (resp. lighter) shaded regions correspond to one standard deviation (resp. two standard deviations).
As expected, the MDP method outperforms the MPC one.
As the confidence bounds around the transition probabilities increase, the robust MDP solution converges to that of the MDP one.

\begin{figure}[tb!]
\centering
\hspace{-10mm}
\includegraphics[scale=.25]{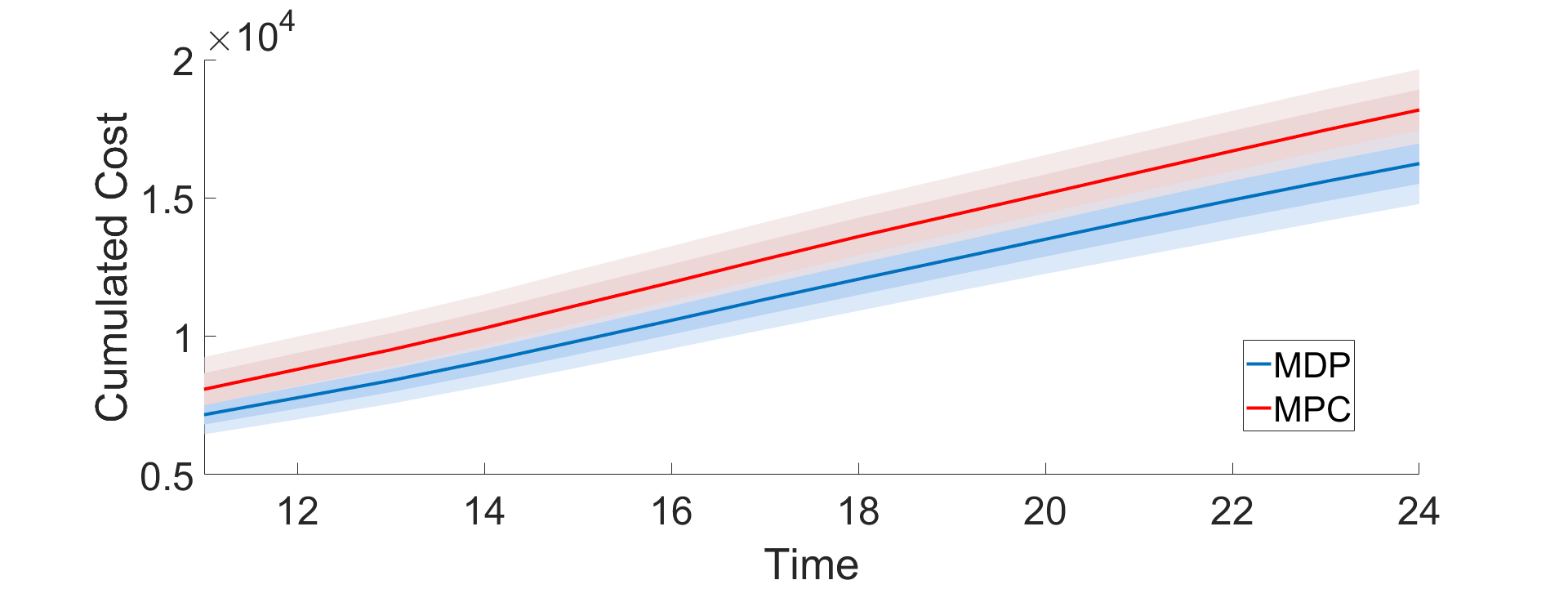}
\caption{Cumulated costs for the MDP and MPC approaches for the last 12 hours of the day.}
\label{fig:mdpmpc}
\end{figure}

The backward-induction algorithm is relatively efficient. 
For solving the previous instance, it takes less than 10 minutes.
We also run our MDP approach for different number of servers and observe that computation times roughly scale linearly in the total number of servers.


\section{Discussions}
In this section, we present the extension of our previous result and algorithm to the infinite horizon and discuss about potential ideas to tackle problems that would be too large for our Algorithm~\ref{alg:bi} to handle.

\subsection{Infinite Horizon}

For the sake of completeness, we also present the infinite horizon case.
However, in the robust case, as the transition function is not known perfectly, a better approach in our opinion would not be to solve for the optimal robust stationary policy, but to use a receding-horizon control kind of approach \cite{SurovikScheeres15,VanOtterloWiering12}. 

In the infinite horizon, we assume that the MDP model is stationary (i.e., $p_t = p$ and $c_t = c$).
If for every finite horizon, an optimal threshold policy exists, this is also true at the limit.
Moreover, the optimal threshold decision rule satisfying the unique limit value of (\ref{eq:optrobust}) defines an optimal stationary threshold policy.
\begin{cor}
If $c_j(x, \lambda, Q)$ is a convex function of $(x, Q)$ for all $\lambda$, then the MDP formulated above admits for the infinite horizon case an optimal policy that is a stationary threshold policy.
\end{cor}
%

Running Algorithm~\ref{alg:bi} until convergence would allow to compute a (near-)optimal stationary threshold policy.
However, it may be impractical if the number of server types/blocks and/or the number of job classes are high.
In that case, a Monte-Carlo based method (Algorithm~\ref{alg:mc}), which was proposed in \cite{PetrikWu15}, may be used.
In that algorithm, $\hat v_n\big((\tau^k_*(\theta))_{\substack{\theta\in\Theta\\ k\in\mathcal K}}\big)$ is computed by sampling and is an estimator of the sum 
$\sum_{(x, \lambda, \theta)\in S} v^*(x, \lambda, \theta)$.
Parameter $n$ is the number of simulations performed for computing $\hat v_n$ and parameter $\epsilon$ is the required precision for the final solution.
In line~\ref{alg:mc:opt} of Algorithm~\ref{alg:mc}, the best threshold $\tau^k(\theta)$ for fixed $\theta$ and $k$ is searched for, while keeping all other thresholds $\tau^{k'}_*(\theta')$ for $\theta'\neq\theta$ and $k'\neq k$ constant.

\begin{algorithm}[tbp]
 \caption{Monte-Carlo$(\langle S, A, p, c \rangle, \gamma, n, \epsilon)$}
 \label{alg:mc}
 \begin{algorithmic}[1]
 \State $t \gets 1$
 \State $v_0 \gets +\infty$
 \State initialize $\tau^k_*(\theta)$ to random values for all $\theta$ and $k$
 
 \Repeat
   \For{$\theta\in\Theta$}
   \For{$k\in\mathcal K$}
    \State $\tau^k_*(\theta) \gets \displaystyle\operatorname*{arg\,min}_{\tau^k(\theta)} \hat v_n\big((\tau^{k'}_*(\theta'))_{\substack{\theta'\neq\theta\\ k'\neq k}}, \tau^k(\theta)\big)$ \label{alg:mc:opt}
    \State $v_t \gets \hat v_n\big((\tau^k_*(\theta))_{\substack{\theta\in\Theta\\ k\in\mathcal K}}\big)$
    \EndFor
    \EndFor
    \State $t \gets t +1$
  \Until{$|v_{t-1} - v_t| \le \epsilon$}
 \State \Return $\tau^*$
 \end{algorithmic}
\end{algorithm}

\subsection{Stronger Assumptions}

We now present two cases with stronger assumptions that make the MDP model more compact, enabling faster solution times.

Recall $[I]$ represent the set of different server types.
In the first case, we assume that all the $c^+_b$ and $c^-_b$ are constant and equal to $c^+$ and $c^-$ respectively and the electricity prices are constant, but possibly different in each server block.
Then the MDP cost function does not depend on time anymore and the MDP model may be simplified into $\langle S, A, (p_t)_{t\ge 1}, c \rangle$:
\begin{itemize}
  \item $S = \Omega_y \times \Lambda \times \Theta$ where $\Omega_y = \prod_{i\in[I]} [N_i]$  with $N_i$ the total aggregated number of servers of type $i$,
  \item $A = \Omega_y$, where an action represents a vector of numbers of servers that will be on for each server type during the next time slot,
  \item $p_t : S\times A\times S\to [0, 1]$ is defined as follows: $\forall s = (y, \lambda, \theta)\in S, s' = (y', \lambda', \theta')\in S$,
  \begin{align}
  p_t\big(s, a, s'\big) &= \mathbb I(a=y') P_t(\lambda' | \theta') P_t(\theta' | \theta)
  \end{align}
  \item $c : S\times A\to \mathbb R$ is a cost function defined as follows:
   \begin{align}
  c\big((y,\lambda, \theta), a\big) = &c^{QoS}(x^*(y), \lambda) + c^E(x^*(y), a) 
  \end{align}
  where $x^*(y)$ is the solution of linear program:
  \begin{align}
  \min &\sum_{b\in [B]} e_b(x_b) \\
  \mbox{s.t. ~ } ~& x \in\Omega_x \nonumber\\
  & y = \mbox{aggregation of } x \mbox{ by server type } \nonumber
  \end{align}
  where $e_b(x_b)$ is the energy consumption cost of leaving $x_b$ servers on in block $b$ during one time slot.
\end{itemize}
As the costs related to $c^+$, $c^-$ and electricity prices are constant, keeping track of $x$ (i.e., the number of servers on in each block) is not needed as it is always better, when required, to switch on servers where electricity costs are lowest first and to switch off servers where those costs are highest first.
For this reason, knowing the current number of servers on by server types is sufficient.
This MDP could be used for modeling situations where the constant cost assumption does not hold by using for instance average prices over a window of time composed of many time steps, or when the duration between two time steps and the horizon are small.

In the second case, we assume the costs of switching on and off are negligible (i.e., $c^+_b = c^-_b = 0$) compared to the other costs.
Then the MDP model can be rewritten as  $\langle S, A, (p_t)_{t\ge 1}, (c_t)_{t\ge 1} \rangle$:
\begin{itemize}
  \item $S = \Omega_y \times \Lambda\times\Theta$,
  \item $A = \Omega_y$, 
  \item $p_t : S\times A\times S\to [0, 1]$ is defined as follows: $\forall s = (y, \lambda, \theta)\in S, s' = (y', \lambda', \theta')\in S$,
  \begin{align}
  p_t\big(s, a, s'\big) &= \mathbb I(a=y') P_t(\lambda' | \theta') P_t(\theta' | \theta)
  \end{align}
  \item $c_t : S\times A\to \mathbb R$ is a cost function defined as follows:
   \begin{align}
  c_t\big((y,\lambda,\theta), &a\big) = c^{QoS}(x^*(y), \lambda) + c^E_t(x^*(y), a) 
  \end{align}
  where $x^*(y)$ is the solution of linear program:
  \begin{align}
  \min & \sum_{b\in [B]} e_b(t, x_b) \\
  \mbox{s.t. ~ } ~& x \in\Omega_x \nonumber\\
  & y = \mbox{aggregation of } x \mbox{ by server type } \nonumber
  \end{align}
\end{itemize}
With the assumption that there is no cost for switching on or off servers, again, keeping track of the number of servers on for each server type is also sufficient as it is possible to adapt to changing electricity prices at every time step by switching off servers where electricity prices are high and switching on servers where electricity prices are low.
This second MDP could be used as an approximation technique when the duration between two time steps is large, which would justify to neglect the costs of switching on or off servers with respect to the cost of leaving a server on.

In both cases, Algorithm~\ref{alg:bi} can be adapted by replacing $\Omega_x$ by $\Omega_y$.
Besides, note that the two previous model simplifications may naturally also be applied in the infinite horizon case.

\section{Conclusion}

%

In this paper, we proposed a robust MDP that models a more realistic data center control problem with heterogeneous servers and different job classes.
The resulting MDP has some structure that one can exploit in order to solve it efficiently.
We proved under a convexity assumption that an optimal threshold policy exists and propose an algorithm for computing it.

As future work, we may as in \cite{ParoliniSinopoliKrogh08} integrate the control of the cooling system in our model and investigate efficient solving methods.
Currently, the solving method proposed in \cite{ParoliniSinopoliKrogh08} is not scalable to medium or large-sized data centers.
Another direction would be to relax some of our assumptions.
Indeed, in practice, the dynamics of parameters $\theta$ may not simply be a Markov chain and may follow some pattern (e.g., rush-hour may happen roughly during the same hours).
Besides, the parameter may not be observable.
In that case, a model like Hidden-Semi-Markov-Mode MDP proposed by \cite{HadouxBeynierWeng14} could be exploited.

\bibliographystyle{abbrv}
\bibliography{biblio141117}

\end{document}